\documentclass[10pt,journal,compsoc]{IEEEtran}
\ifCLASSOPTIONcompsoc
  \usepackage[nocompress]{cite}
\else
  \usepackage{cite}
\fi

\ifCLASSINFOpdf
  \usepackage[pdftex]{graphicx}
\else
\fi

\usepackage{amsmath}
\usepackage{array}

\usepackage{algorithm}
\usepackage{amssymb}
\usepackage{amsthm}
\usepackage{float}  
\usepackage[T1]{fontenc}
\usepackage{geometry}
\usepackage{listings}
\usepackage{mathtools}
\usepackage{graphicx}
\usepackage{xspace}
\xspaceaddexceptions{]\{\}}
\usepackage{svg}
\usepackage{verbatim,listings}
\usepackage{float}
\usepackage[T1]{fontenc}
\usepackage{booktabs}

\usepackage[ampersand]{easylist}
\ListProperties(Indent=4ex)

\usepackage[noend]{algpseudocode}

\usepackage{complexity}
\newcounter{mycounter}
\newenvironment{noindlist} {\begin{list}{\arabic{mycounter})~~}{\usecounter{mycounter} \labelsep=0em \labelwidth=0em \leftmargin=0em \itemindent=0em}} {\end{list}}
\newenvironment{noindqlist} {\begin{list}{$\ket{\arabic{mycounter}}$~~}{\usecounter{mycounter} \labelsep=0em \labelwidth=0em \leftmargin=0em \itemindent=0em}} {\end{list}}
\newenvironment{qlist} {\begin{list}{$\ket{\arabic{mycounter}}$~~}{\usecounter{mycounter} \labelsep=0em \labelwidth=0em}} {\end{list}}

\newtheorem{defn}{Definition}
 \newtheorem{lemma}{Lemma}

\newcounter{TableCounter}

\graphicspath{{graphics/}}
\makeatletter
\def\input@path{{graphics/}}
\makeatother

\providecommand{\csword}[1]{\ensuremath{\text{\ComplexityFont{#1}}}\xspace}

\providecommand{\jacobi}[2]{\left( \displaystyle\frac{#1}{#2} \right)}
\providecommand{\Z}{\mathbb{Z}}
\providecommand{\Zp}{\Z/p\Z}
\providecommand{\Qbar}{\ensuremath{\overline{\mathbb{Q}}}}

\providecommand{\ket}[1]{\left| #1 \right\rangle}

\providecommand{\qpic}{$\langle\textrm{q}|\textrm{pic}\rangle$\xspace}
\providecommand{\NOT}{\csword{NOT}}
\providecommand{\CNOT}{\csword{CNOT}}
\providecommand{\CNOTs}{\csword{CNOT}{}s\xspace}
\providecommand{\CZ}{\csword{CZ}}
\providecommand{\CCZ}{\csword{CCZ}}
\providecommand{\XOR}{\csword{XOR}}
\providecommand{\xT}{\csword{T}}

\providecommand{\xPpoly}{\csword{P/poly}\xspace}
\providecommand{\xP}{\csword{P}}
\providecommand{\RP}{\csword{RP}}
\providecommand{\BPP}{\csword{BPP}}
\providecommand{\BQP}{\csword{BQP}}
\providecommand{\EQP}{\csword{EQP}}
\providecommand{\EQPC}{\ensuremath{\csword{EQP}_\mathbb{C}}\xspace}
\providecommand{\EQPK}{\ensuremath{\csword{EQP}_K}\xspace}
\providecommand{\EQPQ}{\ensuremath{\csword{EQP}_{\Qbar}}\xspace}
\providecommand{\QNR}{\ensuremath{\csword{QNR}}\xspace}

\providecommand{\jacobi}[2]{\left( \displaystyle\frac{#1}{#2} \right)}
\providecommand{\Z}{\mathbb{Z}}
\providecommand{\Zp}{\Z/p\Z}

\providecommand{\ket}[1]{\left| #1 \right\rangle}

\providecommand{\qpic}{$\langle\textrm{q}|\textrm{pic}\rangle$\xspace}
\providecommand{\NOT}{\ComplexityFont{NOT}\xspace}
\providecommand{\CNOT}{\ComplexityFont{CNOT}\xspace}
\providecommand{\CZ}{\ComplexityFont{CZ}\xspace}
\providecommand{\CCZ}{\ComplexityFont{CCZ}\xspace}

\providecommand{\XOR}{\ComplexityFont{XOR}\xspace}

\usepackage{geometry}

\providecommand{\qq}{Q_{17}}
\providecommand{\ii}{\left[\jacobi{x}{17}=-1\right]}
\providecommand{\iii}{\left[\left(\frac{x}{17}\right)=-1\right]}
\begin{document}

\title{Evaluating NISQ Devices with Quadratic Nonresidues}

\author{Thomas~G.~Draper
\IEEEcompsocitemizethanks{\IEEEcompsocthanksitem T. Draper is with the Center for Communication Research at La Jolla,
4320 Westerra Court, San Diego, CA, 92121.\protect\\
E-mail: tdraper@ccrwest.org}
}

%
\IEEEtitleabstractindextext{%
\begin{abstract}
  We propose a new method for evaluating NISQ devices.
  This paper has three distinct parts.
  First, we present a new quantum algorithm that solves a two hundred year old problem of finding quadratic nonresidues (QNR) in polynomial time.
  We show that QNR is in \EQPC, while it is still unknown whether QNR is in \P.
  Second, we present a challenge to create a probability distribution over the quadratic nonresidues.
  Due to the theoretical complexity gap, a quantum computer can achieve a higher success rate than any known method on a classical computer.
  A device beating the classical bound indicates quantum advantage or a mathematical breakthrough. 
  Third, we derive a simple circuit for the smallest instance of the quadratic nonresidue test and run it on a variety of currently available NISQ devices.
  We then present a comparative statistical evaluation of the NISQ devices tested.

\end{abstract}
}


\maketitle

\IEEEdisplaynontitleabstractindextext

\IEEEpeerreviewmaketitle

\IEEEraisesectionheading{\section{Introduction}\label{sec:introduction}}

\IEEEPARstart{Q}{uantum} advantage is an elusive result to establish for emerging NISQ devices.
Various approaches restrict communication, circuit depth, or circuit width to show a quantum advantage~\cite{Centrone_2021, Kumar_2019,Arrazola_2018,Bravyi_2018,aaronson2016complexitytheoretic,46227}.

We seek to demonstrate quantum advantage via an alternate route.
We build a test based on a centuries-old math problem which lies at the very heart of number theory: quadratic reciprocity.
Instead of attacking a problem too hard for modern supercomputers, we invite NISQs to run an algorithm to generate a probability distribution, using limited resources, that a classical computer cannot.

In particular, given a prime $p$, the challenge is to create a probability distribution only over the quadratic nonresidues of $p$.
As will be shown, a quantum computer can create a uniform superposition of the quadratic nonresidues modulo $p$ in polynomial time.
Proving any task is hard, is a challenging task in and of itself.
We avoid the difficulty of proving the task to be hard by noting that since the time of Gauss\cite{gauss_qnr}, mathematicians have tried and failed to find an algorithm that produces quadratic nonresidues in \xP.
Thus, this problem is considered hard for the same reason factoring is considered hard: The problem is very old, and lots of talented mathematicians have failed to solve it.
Not being able to find quadratic nonresidues in polynomial time is often referred to as an ``embarrassment'' in number theory, since we can easily find quadratic nonresidues in random polynomial time (\RP).

To turn this quantum/classical difference into a test for NISQs, we compare the success rate of algorithms that are allowed to compute a single Jacobi symbol (used to identify quadratic residues and nonresidues).
With this restriction, a perfect quantum computer succeeds with probability $1$, whereas a classical computer succeeds with probability $\frac{3}{4}$.
If a NISQ succeeds more than $\frac{3}{4}$ of the time, we have evidence of quantum advantage.

Favorable aspects of this approach include:
\begin{itemize}
  \item {\bf Implementation agnostic.} Quantum algorithms that are tailored to what a particular NISQ can do well might not run as well on a different NISQ. Having a math-inspired algorithm doesn't intentionally favor any particular architecture.
  \item {\bf Infinite tests.} This approach provides an endless supply of problems. Little effort is needed to find new problems and prove they are hard.
  \item {\bf Ease of execution.} As we will see, a NISQ can have as few as four qubits and still run a meaningful test. Evaluation of a NISQ can be done on a wide variety of sizes, giving a better insight to how the performance degrades as size increases.
  \item {\bf Better comparison.} In comparing two NISQs, a performance certificate would consist of the prime targeted, the circuit used, and the score obtained. If a more efficient circuit for a particular prime is found, all parties can run the test with the improved circuit. That way the test remains a test of the hardware.
\end{itemize}

\subsection{Overview of topics}

\noindent \QNR is in \EQPC
\begin{itemize}
  \item {\bf Quadratic nonresidues:} Provide basic number theory background.
  \item {\bf Grover's algorithm:} Revisit how a Grover iteration inverts about the mean.
  \item {\bf Amplifying quadratic nonresidues:} Show how complex rotations on quadratic nonresidues before a Grover iteration creates a perfect superposition.
  \item {\bf Algorithmic complexity:} Show that the algorithm lies in \EQPC.
\end{itemize}

\noindent Design a \QNR based test
\begin{itemize}
  \item {\bf General test for NISQ devices:} Algorithm to generate circuit test for arbitrary primes where $p\equiv 1\bmod 8$.
  \item {\bf Fermat primes:} Show that a simpler test is available for Fermat primes.
  \item {\bf Quantum circuit:} Construct a circuit for the smallest interesting Fermat prime ($p=17$).
  \item {\bf The \xPpoly problem:} Examine the pitfalls of over-optimizing circuits.
\end{itemize}

\noindent Results of running the \QNR test on current NISQs
\begin{itemize}
  \item {\bf Test success rate on current NISQs:} Show the success rates of currently available NISQs.
  \item {\bf Test uniformity on current NISQs:} Evaluate the likelihood of the samples coming from a uniform distribution.
\end{itemize}

\section{Quadratic Nonresidues}
\begin{defn}
  Let $a,p \in \Z$ where $\gcd(a,p)=1$.
  If $x^2 \equiv a \bmod{p}$ has a solution, then $a$ is a \emph{quadratic residue} modulo $p$.
  Otherwise, $a$ is a \emph{quadratic nonresidue} modulo $p$.
\end{defn}

\begin{defn}
  Let $p$ be an odd prime and $a\in \Z$.
  The {\bf Legendre symbol} is defined as
  \begin{equation}
    \jacobi{a}{p}=
  \begin{cases}
    1&\text{if }a\text{ is a quadratic residue}\bmod p,\\
    -1&\text{if }a\text{ is a quadratic nonresidue}\bmod p,\\
    0&\text{if }a\equiv 0\bmod{p}.\\
  \end{cases}
\end{equation}
\end{defn}

Finding a quadratic nonresidue modulo a prime $p$ is an easy task.
A random choice of $x \in \{1, \ldots , p-1\}$ has a 50\% chance of being a quadratic nonresidue.
If $p$ is an $n$-bit prime, the Legendre symbol 
can be computed using $\log n$ multiplications, and thus in $O(M(n)\log
n)=O(n\log^2 n)$ time ~\cite{DBLP:journals/corr/abs-1004-2091,harvey:hal-02070778}.
The result certifies whether or not $x$ is a quadratic nonresidue.
Surprisingly, there is no known method for generating a quadratic nonresidue in polynomial time.
Assuming the Generalized Riemann Hypothesis, there exists a quadratic nonresidue less than $O(\log^2 p)$~\cite[p. 34]{1993--cohen}, that can therefore be found in $O(n^3\log^2 n)$ time by incremental search.
This paper presents a new quantum algorithm that finds a quadratic nonresidue in $O(n\log^2 n)$ time, independent of the Riemann Hypothesis.

\subsection{Notation and terminology}
The \emph{Jacobi symbol} and \emph{Kronecker symbol} are both generalizations of the Legendre symbol, and share the same notation. Efficient calculations of the Legendre symbol often use these more general forms, and they will be assumed whenever necessary in this paper.

Given a logical statement $Q$ on some number of variables $x_0, x_1,\ldots$, we will use the notation $[Q(x_0, x_1, \ldots)]$ for the Boolean function with inputs $x_0, x_1, \ldots$ whose output is $1$ or $0$, dependent on whether $Q(x_0, x_1,\ldots)$ is True $(1)$ or False $(0)$.

For an integer variable $x$, the bit variables $x_i$ will be defined from the binary form of $x=x_n\cdots x_2x_1x_0$, with $x_0$ being the least significant bit. Using this notation, we can express the parity function as $[x\mbox{ is odd}]=x_0$.


\subsection{Basic number theory for quadratic residues}
Recall the following facts about quadratic residues for an odd prime $p$:
\begin{equation}\label{pos}
    \left|\left\{a\in\Zp : \jacobi{a}{p}=1\right\}\right|= \frac{p-1}{2}
  \end{equation}
\begin{equation}\label{neg}
    \left|\left\{a\in\Zp : \jacobi{a}{p}=-1\right\}\right|= \frac{p-1}{2}
  \end{equation}
  \begin{equation}\label{qnr1}
    \jacobi{-1}{p}=(-1)^{\frac{p-1}{2}}=
      \begin{cases}
        1 & \text{if } p\equiv 1\bmod{4}\\
        -1 & \text{if } p\equiv 3\bmod{4}\\
      \end{cases}
  \end{equation}
  \begin{equation}\label{qnr2}
    \jacobi{2}{p}=(-1)^{\frac{p^2-1}{8}}=
      \begin{cases}
        1 & \text{if } p\equiv 1,7\bmod{8}\\
        -1 & \text{if } p\equiv 3,5\bmod{8}\\
      \end{cases}
  \end{equation}

Using equations (\ref{qnr1}) and (\ref{qnr2}), we see that $-1$ or $2$ is a quadratic nonresidue unless $p\equiv 1\bmod{8}$.
Thus the primes for which finding a quadratic nonresidue is nontrivial are congruent to $1$ modulo $8$.

\begin{lemma}
  Let $p$ be prime such that $p \equiv 1 \bmod{4}$.
  If we take $1,2,\ldots,p-1$ as our nonzero congruence class representatives,
  then half of the quadratic nonresidues are even and half are odd.
\end{lemma}

\begin{proof}
  By (\ref{qnr1}), $-1$ is a quadratic residue when $p \equiv 1 \bmod{4}$.
  Thus if $x$ is a quadratic nonresidue, then $-x\equiv p-x\bmod{p}$ is also a quadratic nonresidue.
  Since $p$ is odd, $x$ and $p-x$ have different parities, and $x\neq p-x$.
  Therefore, every odd quadratic nonresidue has a unique matching even nonresidue.
  Using (\ref{neg}), the number of nonresidues is $\frac{p-1}{2}$.
  Thus the number of odd (respectively even) nonresidues is $\frac{p-1}{4}$.
\end{proof}


It should be noted that there are additional number theoretic techniques that can be used to help us find quadratic nonresidues.
\begin{lemma}\label{mod_lemma}
 Let $p\equiv 1 \bmod 4$ be prime, and let $q$ be an odd prime less than $p$. Then $\jacobi{q}{p}=\jacobi{p\bmod q}{q}$.
\end{lemma}

\begin{proof}
 Using Gauss's quadratic reciprocity and the fact that $p\equiv 1 \bmod 4$, we see that
  \[ \jacobi{q}{p}=(-1)^{\frac{p-1}{2}\cdot\frac{q-1}{2}}\jacobi{p}{q}=\jacobi{p}{q}=\jacobi{p\bmod q}{q}.\]
\end{proof}

Lemma~\ref{mod_lemma} implies that $q$ is a quadratic nonresidue of $p$ if and only if $p\bmod q$ is a quadratic nonresidue of $q$. This can be used to find classes of primes $p$ where $3,5,7,$ or some other small prime is a quadratic nonresidue.
For example, if $p\equiv 1\bmod4$ and $p\equiv 2\bmod 3$, then $3$ is a quadratic nonresidue modulo $p$.
Results like this would improve a classical search, but we are going to disallow
this approach and other number theoretic boosts, since they fail to solve the problem in all cases, and reduce the number of interesting primes to use in our classical/quantum comparison.

\section{Grover's Algorithm}
In 1996, Lov Grover~\cite{Grover96} devised an unstructured search algorithm using amplitude amplification. In the case of finding a single marked entry out of $N$, Grover's algorithm reduces the number of black box queries from $O(N)$ to $O(\sqrt{N})$. When the acceptable answer space has size $k$, the expected number of black box queries drops from $O(N/k)$ to $O(\sqrt{N/k})$.

Each step of Grover's algorithm involves flipping the signs of the amplitudes of the marked states and then inverting about the mean.

\subsection{Inversion about the mean}
Although Grover's algorithm uses only real amplitudes, the inversion about the mean holds for complex numbers and is a 2D inversion about a point.
Let $\bar{\alpha}$ be the mean of the amplitudes of a given state.
\[\bar{\alpha} = \frac{1}{2^n}\sum_{x=0}^{2^n-1}\alpha_x.\]
The inversion step of Grover's algorithm proceeds as follows.
\begin{qlist}
  \item Initial state before inversion
    \[ \displaystyle\sum_{x=0}^{2^n-1}\alpha_x\ket{x}
    \]
  \item Quantum Hadamard Transform
    \[ \displaystyle\sum_{x=0}^{2^n-1}\frac{\alpha_x}{2^{n/2}}\sum_{y=0}^{2^n-1}(-1)^{x\cdot{}y}\ket{y}
    \]
  \item Negate the phase of the zero state

    \[ {\left(\sum_{x=0}^{2^n-1}\frac{\alpha_x}{2^{n/2}}\sum_{y=0}^{2^n-1}(-1)^{x\cdot{}y}\ket{y}\right)}
    { -\sum_{x=0}^{2^n-1}\frac{\alpha_x}{2^{n/2}}\cdot{}2\ket{0}}
    \]
  \item Quantum Hadamard Transform
  \[
    \displaystyle\sum_{x=0}^{2^n-1}\alpha_x\ket{x}-2\bar{\alpha}\sum_{x=0}^{2^n-1}(-1)^{x\cdot{}0}\ket{x}
    \]
  \item Change global phase by $-1$ (NOP)
  \[
    \displaystyle\sum_{x=0}^{2^n-1}\left(2\bar{\alpha}-\alpha_x\right)\ket{x}
    \]
\end{qlist}

Thus inversion about the mean maps each amplitude
\[\alpha_x \mapsto 2\bar{\alpha}-\alpha_x.\]
In using an inversion about the mean to find quadratic nonresidues, we will use complex amplitudes.

\section{Amplify Quadratic Nonresidues}
It is well known that a single Grover iteration can amplify one of four states to probability one of observation.
In general, Grover iterations only use real amplitudes and do not tend to perfectly amplify the target states.
By using complex amplitudes and balancing the even and odd nonresidues with each other, we can ensure perfect amplification after a deterministic number of steps.

Let $n$ be the least integer such that $2^n>p$ and let $N=2^n$.
By varying the Grover step slightly, we may generate quadratic nonresidues for $p\equiv 1\bmod{8}$.
Instead of negating the phase, we will rotate the complex phase of the even and
odd quadratic nonresidues in such a way that the mean moves to a position that
will send all non-nonresidues ($\ket{0}$, residues and values greater than or equal to $p$  but less than $N$) to zero upon inversion.
It is important to restrict the rotation to quadratic nonresidues less than $p$.
The distribution of residues and nonresidues from $p$ to $2^n-1$ is unknown and is unlikely to be equally partitioned.

Note that after the inversion, all nonresidues will have equal probability of being seen, and that all are equally likely to be observed.

\subsection{Example: Amplify the nonresidues for \protect{\boldmath{$p=41$}}, \protect{\boldmath{$N=64$}}}
In this example, we will use QNR to denote {\it quadratic nonresidues}.
\begin{qlist}
  \item Initialize all $64$ values ($\ket{0},\ldots ,\ket{63}$) in an equal superposition.
    Since all values have an amplitude of $1/8$, the target mean will be $1/16$.
    Since $p=41$, there are $20$ QNRs, of which $10$ are even and $10$ are odd.

\begin{figure}[h]
\centering
  \includegraphics[width=2.5in]{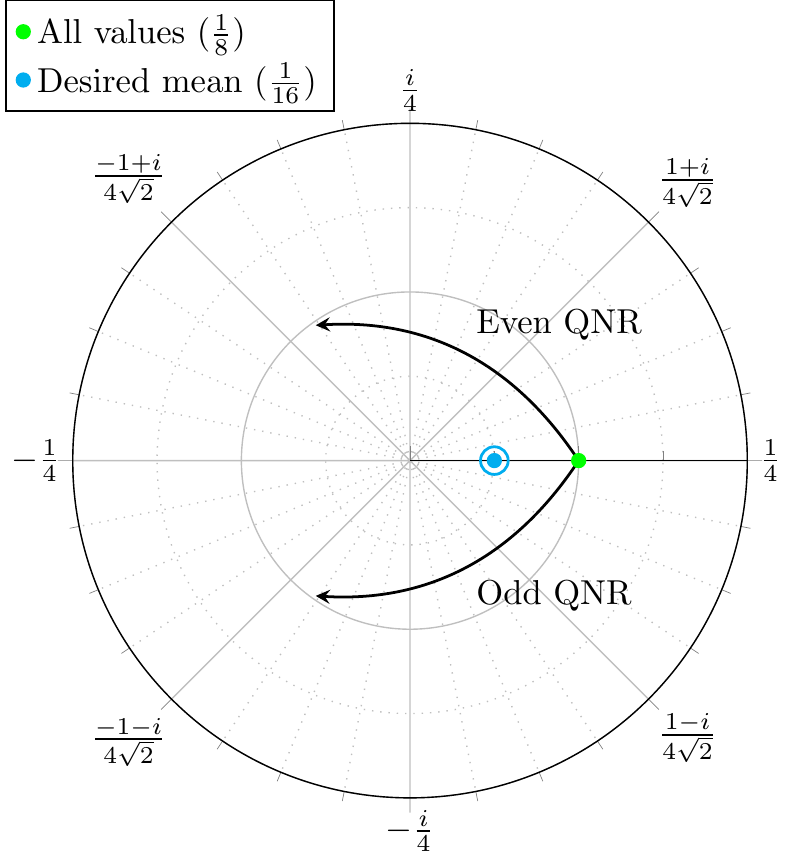}
    \caption{Rotate QNRs, even by $\theta$ and odd by $-\theta$.}
\label{p41_1}
\end{figure}

\item Let $\theta=\hbox{arccos}(-3/5)$. The even QNRs will be rotated by $\theta$ and the odd QNRs will be rotated by $-\theta$ (figure~\ref{p41_1}).
  After rotation, the even/odd QNR amplitudes are $-\frac{3}{40}\pm\frac{i}{10}$.
  In calculating the mean, the imaginary components of the even and odd
  quadratic nonresidues will exactly cancel each other.
  The remaining average on the real line is thus
  \[\frac{1}{64}\left(20\cdot\left(-\frac{3}{40}\right)+44\cdot\frac{1}{8}\right)=\frac{1}{16}.\]

\begin{figure}[h]
\centering
  \includegraphics[width=2.5in]{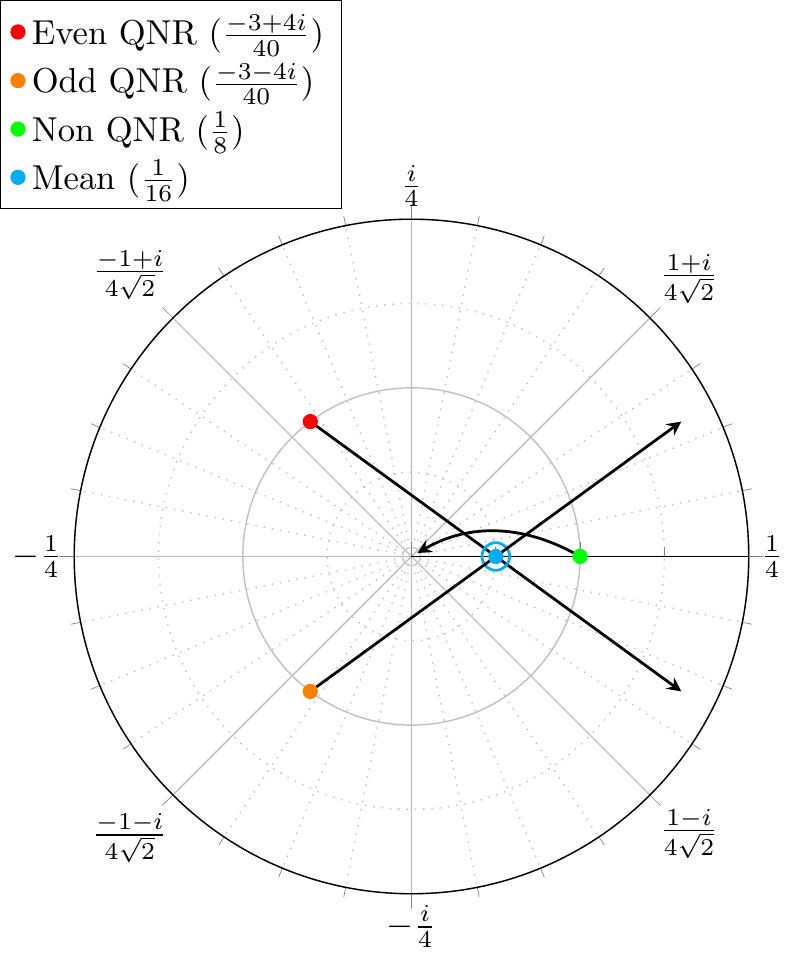}
    \caption{Invert about the mean.}
\label{p41_2}
\end{figure}
\item Invert all complex amplitudes about the mean. Each amplitude $a$ will invert to $2(\frac{1}{16})-a$ (figure~\ref{p41_2}). In particular, $\frac{1}{8}$ inverts to zero, and the other amplitudes do not.

\begin{figure}[h]
\centering
  \includegraphics[width=2.5in]{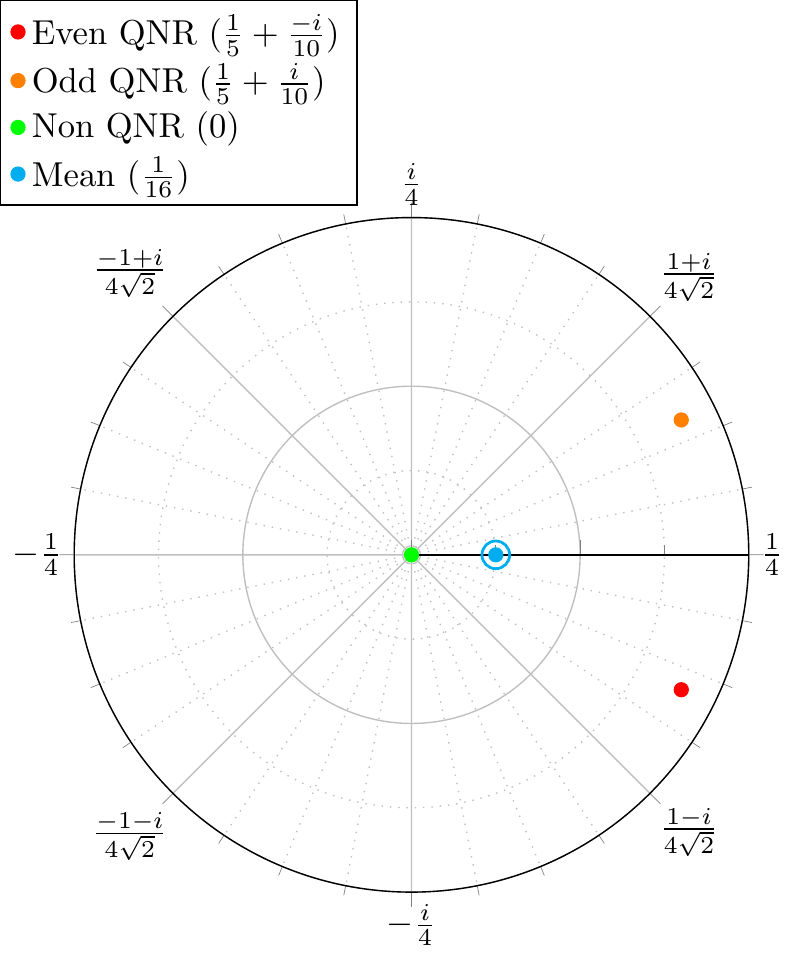}
    \caption{Final position where only QNRs have nonzero (albeit complex) amplitude.}
\label{p41_3}
\end{figure}
\end{qlist}

After inversion, all of the nonresidues are equally likely to be observed
(figure~\ref{p41_3}).
Although their phases are different, the magnitudes of their amplitudes are identical, and the magnitude dictates the probability of observation.

\subsection{Computing the angle of rotation}
To find the correct angle of rotation for a given prime, we need to find the angle which moves the mean of the amplitudes to half of the amplitude of the equal superposition.
Since the imaginary component of the mean will be zero, we seek the real component $a$ for the $(p-1)/2$ QNRs. 
\[ a\left(\frac{p-1}{2}\right) + 1\cdot\left(2^n-\frac{p-1}{2}\right)= \frac{1}{2}\cdot 2^n\]
Solving for $a$ gives
\[ a = 1- \frac{2^n}{p-1} \]
and thus the angle of rotation is
\begin{align}
  \label{theta}
\theta=\mbox{arccos}\left(1-\frac{2^n}{p-1}\right). 
\end{align}

\subsection{Amplitude before and after inversion}
In an equal superposition of $N=2^n$ states, each state has amplitude $1/\sqrt{N}$. 
After a rotation by $\pm\theta$, the quadratic nonresidues have amplitude
\[ \frac{1}{\sqrt{N}}e^{\pm i\theta}. \]

  Since $\theta$ was chosen to move the amplitude mean from $\frac{1}{\sqrt{N}}$ to $\frac{1}{2\sqrt{N}}$, inversion about the mean maps the nonresidue amplitudes to 
  \[2\frac{1}{2\sqrt{N}}-\frac{1}{\sqrt{N}}e^{\pm i \theta} =\frac{1}{\sqrt{N}}\left(1-e^{\pm i \theta}\right). \]
Since over half of the amplitudes were sent to $0$ in the inversion, the nonresidue amplitude length must have increased.
The new amplitudes of the $\frac{p-1}{2}$ nonresidues have squared length
\begin{align*}
  &\frac{1}{N}\left((1-\cos\pm\theta)^2+(\sin\pm\theta)^2\right)\\
  =&\frac{2}{N}\left(1-\cos\theta\right)\\
  =&\frac{2}{N}\left(\frac{N}{p-1}\right)\\
  =&\frac{2}{p-1}.\\
\end{align*}

\section{Algorithmic Complexity}
\subsection{Complexity classes}
Early attempts to find a quantum analog for \xP led to the definition of Exact Quantum Polynomial Time (\EQP)~\cite{Bernstein97quantumcomplexity}.
Unfortunately, which algorithms \EQP deemed polynomial time depended on the finite generating set chosen.
In practice, researchers found \BQP, the quantum analog of \BPP, more useful for discussing what a quantum computer can do efficiently.

In an effort to talk about what a quantum computer can do in polynomial time, the complexity class \EQPK was introduced~\cite{Adleman97quantumcomputability}.
\EQPK allows for controlled unitary gates $U$ on a single qubit where the coefficients of $U$ are from a set $K$.
In practice, the interesting cases come when $K$ is infinite.
For finding quadratic nonresidues, we are interested in the case where $K=\mathbb{C}$.
Note that the coefficients for $U$ may be drawn from $\Qbar$ instead of $\mathbb{C}$~\cite{Adleman97quantumcomputability}.
The discrete logarithm problem over $\mathbb{Z}/p\mathbb{Z}$ is in \EQPQ~\cite{Mosca04exactquantum}.

Although using an infinite generating set might seem bad, the Solovay-Kitaev theorem~\cite{dawson2005solovaykitaev} proves the single qubit unitary gates can be approximated with exponential precision in polynomial time from a reasonable finite generating set. A more efficient implementation of this idea can be found in \cite{selinger2012efficient}.

We will show that the quadratic nonresidue algorithm is in \EQPC.

\subsection{QC algorithm for finding QNRs in polynomial time}
\label{qc_alg}
Let $M(n)$ be the time required to multiply two $n$ bit numbers. 
The entire algorithm will have complexity equal to the Jacobi symbol, $O(M(n)\log n)=O(n \log^2 n)$~\cite{DBLP:journals/corr/abs-1004-2091,harvey:hal-02070778}.

Given a prime $p\equiv 1\bmod 8$, we define the following for use in algorithm~\ref{quantum_qnr}.
\[ N/2 < p < N=2^n, \]
\[ \theta=\mbox{arccos}\left(1-\frac{N}{p-1}\right), \]
\[ x_0 \text{ is the low bit of } x, \]
\begin{align*}                                                                                                             \label{indicator}
  f(x)&=\left[\jacobi{x}{p}=-1\text{ and } 0\le x<p\right]\\                                                              &=
  \begin{cases}                                                                                                               1&\text{if }x<p\text{ is a quadratic nonresidue},\\                                                                   0&\text{otherwise}.\\
  \end{cases}
\end{align*}
Even though we want to rotate the even QNRs by $\theta$ and the odd QNRs by $-\theta$, for most quantum computers it is likely to be less work to rotate the odd QNRs by $-2\theta$ and then rotate all QNRs by $\theta$.
This is reflected in the algorithm~\ref{quantum_qnr} and graphically in Figure~\ref{fig:gen_alg}.

\begin{algorithm}
  \caption{
    Quantum polynomial time algorithm for uniformly sampling quadratic nonresidues of a prime $p\equiv 1 \bmod 8$.
    Let $n, N, \theta, x_0 \text{ and } f(x)$ be defined as in section~\ref{qc_alg}.
}
  \label{quantum_qnr}
  \begin{noindqlist}
  \item $[O(n)]$ Apply $H^{\otimes n}$ to \( \ket{0}^{\otimes n} \) (Hadamard transform).
    \[ \frac{1}{\sqrt{N}}\sum_{x=0}^{N-1} \ket{x} \]
  \item $[O(M(n) \log n)]$ Compute Jacobi symbol indicator.
\[\frac{1}{\sqrt{N}} \sum_{x=0}^{N-1} \ket{x}\ket{\left[\jacobi{x}{p}=-1\right]}\]
\item $[O(n)]$ Compute the indicator for $[x<p]$~\cite{comparator}.
\[\frac{1}{\sqrt{N}} \sum_{x=0}^{N-1} \ket{x}\ket{\left[\jacobi{x}{p}=-1\right]}\ket{[x<p]}\]
  \item $[O(1)]$ Rotate odd QNRs less than $p$ by $-2\theta$, conditioned on $[x<p],\left[\jacobi{x}{p}=-1\right],$ and $x_0$.
  \[\frac{1}{\sqrt{N}} \sum_{x=0}^{N-1} e^{-i2\theta f(x)x_0}\ket{x}\ket{\left[\jacobi{x}{p}=-1\right]}\ket{[x<p]}\]
\item $[O(1)]$ Rotate all QNRs less than $p$ by $\theta$, conditioned on $[x<p]$ and $\left[\jacobi{x}{p}=-1\right]$.
  \[\frac{1}{\sqrt{N}} \sum_{x=0}^{N-1} e^{i\theta f(x)(1-2x_0)}\ket{x}\ket{\left[\jacobi{x}{p}=-1\right]}\ket{[x<p]}\]
\item $[O(M(n)\log n)]$ Uncompute indicator functions.
  \[\frac{1}{\sqrt{N}} \sum_{x=0}^{N-1} e^{i\theta f(x)(1-2x_0)}\ket{x}\]
\item $[O(n)]$ Use a Grover step to invert about the mean $\alpha=\frac{1}{2\sqrt{N}}$.
    \[ \frac{1}{\sqrt N}\sum_{x=0}^{N-1} \left(1-e^{i\theta f(x) (1-2x_0)}\right)\ket{x} \]
  \item $[O(n)]$ Observe a quadratic nonresidue modulo $p$.
This observation samples uniformly among all quadratic residues modulo $p$.
  \end{noindqlist}
\end{algorithm}

\begin{figure}[h]
\centering
  \includegraphics[scale=1]{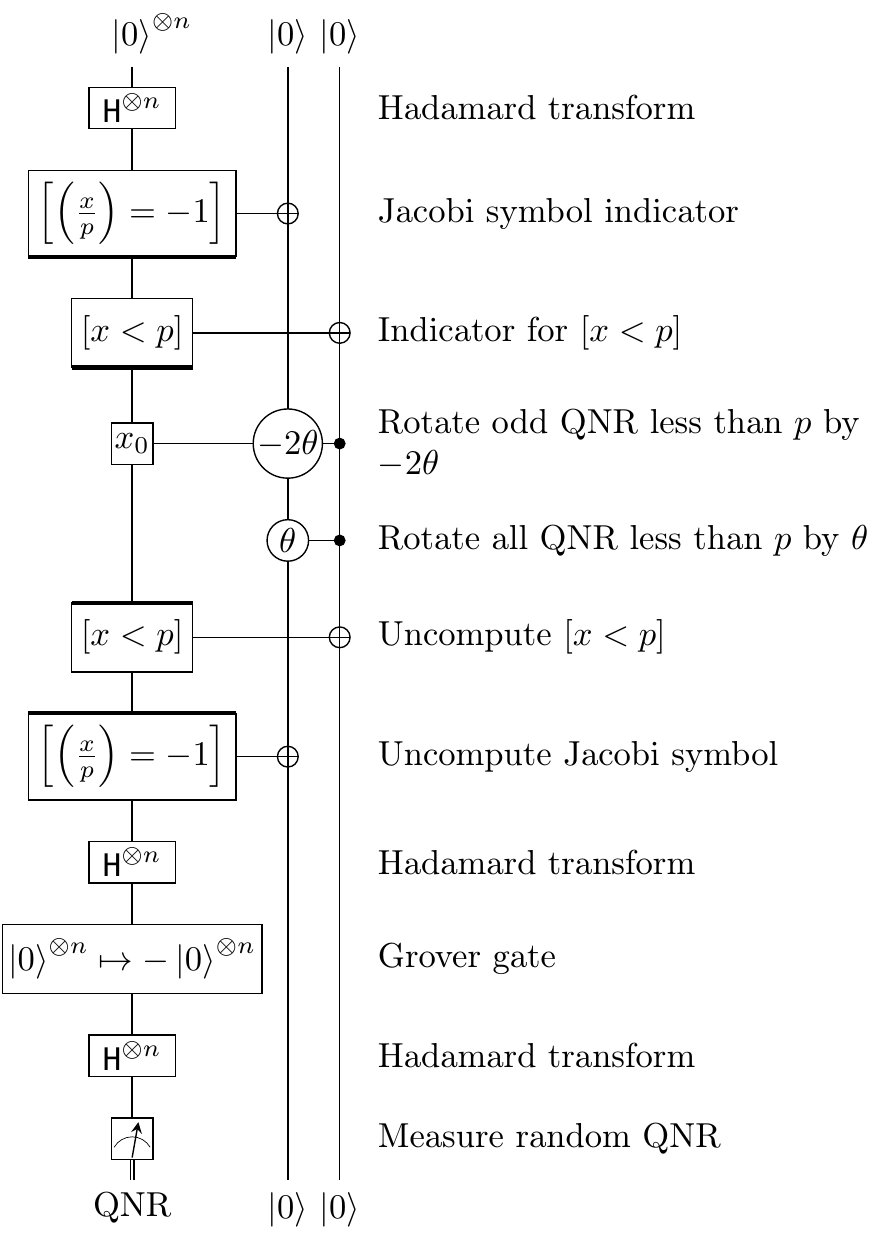}
  \caption{Wire diagram for algorithm~\ref{quantum_qnr} for sampling quadratic nonresidues}
\label{fig:gen_alg}
\end{figure}

\section{A Test for NISQ Devices}
In polynomial time, algorithm~\ref{quantum_qnr} creates a distribution for a given prime $p\equiv 1\bmod 8$ where only the quadratic nonresidues of $p$ are observed and their observation is uniform.
Any classical polynomial time algorithm capable of the same would be an exciting new result in number theory, solving a 200+ year old problem\cite{gauss_qnr}.
We define a test based on the minimal resources needed by a quantum computer to create this distribution and evaluate how well a classical computer could do with similar restrictions.

We then use this as a benchmark to judge how well a particular NISQ performs.
The two properties we will evaluate are:
\begin{itemize}
  \item The frequency of QNRs produced.
  \item The uniformity of the distribution of the observed QNRs.
\end{itemize}
Exceeding the classical bound also provides evidence that the NISQ is doing something quantum.

\subsection{Comparing quantum and classical resources}
How should two computers be given ``equal resources'' if the computers have fundamentally different building blocks?
Essentially, we need to agree on how to compare apples and oranges.
Since a quantum computer (QC) is supposed to be more powerful than a classical computer, some of the equivalences may seem unfair.
But the spirit of this comparison is to evaluate actions on each computer that take roughly the same amount of work.
We propose the following rationale for judging the classical and quantum algorithms presented to have relatively similar difficulty.

\begin{itemize}
  \item {\bf Hadamard vs.~random:}
    $H^{\otimes n}$ can be used to generate an $n$-bit random number on a QC.
    Allowing the classical computer access to a random oracle is roughly comparable to allowing the QC to perform a Hadamard transform in a single time slice, which can be used for the same purpose.
  \item {\bf Jacobi symbol:} This is the bulk of the computation done in both worlds.
    The QC is required to perform the computation reversibly.
  \item {\bf Controlled phase rotation:}
    There isn't a good classical analogue for phase rotations, but they are also a very small part of the quantum computation.
    One controlled and one doubly controlled qubit rotation will be used by the QC.
  \item {\bf Inequality indicator {$[x<p]$}:}
    The quantum algorithm requires this simple computation to target quadratic nonresidues less than $p$.
    Classically, this may be used to reduce a random $n$-bit number to a random integer $0\le x <p$, but it would do so in a non-uniform way.
  \item {\bf Compute rotation angle:} This simple classical computation provides the angle of rotation $\theta$ needed by the QC.
    We allow the classical computer to use this result as well, but it doesn't seem to be useful.
\end{itemize}

Of course, some may feel this isn't a fair comparison and propose something different.
One quibble might be that the NISQ is allowed to perform a computation and an uncomputation of the Jacobi symbol, whereas the classical computer only gets to perform a single computation.
To even things up, one suggestion might be to compare a NISQ to a classical computer with two calls to a Jacobi symbol oracle.
This is a reasonable alternative that will raise the threshold for the NISQ to beat.
For the purposes of this paper, we just compare a single Jacobi symbol calculation.

\subsection{The restricted classical test for finding quadratic nonresidues}
Let $N$ be the smallest power of~2 greater than $p$.
The quantum polynomial time algorithm for finding a quadratic nonresidue modulo $p$ consists of three $n$-qubit Hadamard transforms, one doubly controlled phase rotation, one singly controlled qubit phase rotation, one phase flip of the zero state, and the computation/uncomputation of $[(\frac{x}{p})=-1]$ and $[x<p]$.

This computation is dominated by computing the Jacobi symbol.
The quantum algorithm does not use any further number theoretic facts to find a quadratic nonresidue in polynomial time.

Algorithm~\ref{single} provides a classical approach to finding a quadratic nonresidue restricted to single Jacobi symbol evaluation and a source of randomness.
This simple algorithm will succeed 75\% of the time.
Beating this threshold provides evidence that something interesting is happening on our NISQ, since it is unknown how to asymptotically succeed more than 75\% of the time classically with the given computational restrictions.

\begin{algorithm}
  \caption{(Classical) Randomized polynomial time algorithm for finding a quadratic nonresidue modulo a prime $p$ using a single Jacobi symbol calculation.}
  \label{single}
  \begin{noindlist}
  \item Choose a random  $x$ in $\{1, \ldots , p-1\}$.
  \item Compute the Jacobi symbol for $x$.
  \item If {$\jacobi{x}{p}=-1$}, return $x$.
  \item Else, choose a random  $y$ in $\{1, \ldots , p-1\}$.
  \item Return $y$.
  \end{noindlist}

\end{algorithm}


\subsection{The restricted quantum test for finding quadratic nonresidues}
Given a prime $p\equiv 1\bmod 8$, we define the following for use in algorithm~\ref{quantum_qnr}.
\[ N/2 < p < N=2^n, \]
\[ \theta=\mbox{arccos}\left(1-\frac{N}{p-1}\right), \]
\[ x_0 \text{ is the low bit of } x, \]
\begin{align*}                                                                                                             \label{indicator}
  f(x)&=\left[\jacobi{x}{p}=-1\text{ and } 0\le x<p\right]\\                                                              &=
  \begin{cases}                                                                                                               1&\text{if }x<p\text{ is a quadratic nonresidue},\\                                                                   0&\text{otherwise}.\\
  \end{cases}
\end{align*}

A perfect quantum computer would observe a quadratic nonresidue 100\% of the time, beating the bound of 75\% for the restricted classical computer.
A NISQ exceeding the 50\% bound illustrates that it is doing something better than a completely noisy calculation.
A NISQ exceeding the 75\% bound illustrates a power in excess of what a classical computer can do with a similar resource restriction and provides evidence that the NISQ is doing something quantum.

Even if we decide that 75\% is the incorrect quantum advantage threshold, as long as the quadratic nonresidue problem remains unsolved, there will always be some probabilty of success gap no matter what polynomial resources we allow the classical computer to use.
Thus, the higher a NISQ scores, the greater confidence we have that it is doing something interesting.

\section{Fermat Primes}

Some primes may have an easier $\theta$ rotation for nascent NISQ devices to implement.
We explore one of these prime classes, hoping that future sophisticated NISQ devices will require no such crutch.
Consider primes of the form $2^m+1$. If $m$ has an odd factor $a>1$, then $2^{m/a}+1$ divides $2^m+1$. Thus if $2^m+1$ is prime, $m$ must be a power of~2. Primes of the form $2^{2^n}+1$ are called Fermat primes. Unfortunately, there are only five known Fermat primes:
\begin{align*}
  F_0 &= 2^1 + 1 = 3\\
  F_1 &= 2^2 + 1 = 5\\
  F_2 &= 2^4 + 1 = 17\\
  F_3 &= 2^8 + 1 = 257\\
  F_4 &= 2^{16} + 1 = 65537
\end{align*}

The first two Fermat primes, $3$ and $5$, are not congruent to $1\bmod 8$, and thus not ``hard''.
The remaining Fermat primes, $17$, $257$, and $65537$, provide three excellent QNR test candidates.

Although it is unknown if there are any more Fermat primes, the smallest Fermat number of unknown primality~\cite{Crandall02thetwenty-fourth} is $F_{33}$, which is a staggeringly large $8$-gigabit number. Even if $F_{33}$ were prime (which most mathematicians do not believe to be true), any NISQ capable of running a QNR test on $F_{33}$ would be stretching the definition of ``Intermediate Scale''.

$F_n$ requires $2^n+1$ bits to represent.
Using $N=2^{2^n+1}$, for a Fermat prime, the angle of rotation turns out to be exactly $\pi$, which is easy to implement since it is just a phase flip.

\begin{align}
  \theta&=\mbox{arccos}\left(1-\frac{2^{2^n+1}}{(2^{2^n}+1)-1}\right)\\
  &=\text{arccos}(-1)=\pi.\label{eq:fermat}
\end{align}

Recall that, since $F_n\equiv 1\bmod 4$, $-1$ is a quadratic residue. Since $2^{2^n}\equiv -1$ is a quadratic residue of $F_n$, all quadratic nonresidues for $F_n$ will be less than $2^{2^n}$, and thus all potential quadratic nonresidues require only $2^n$ bits to represent.

Exactly half of the values $0\le x < 2^{2^n}$ will be quadratic nonresidues for $F_n$.
Zero is not typically counted as a regular quadratic residue, but since we are purposely leaving out $-1$, zero will take its place so that exactly half the values are quadratic (non)residues.
Thus for Fermat primes we can save one bit and rotate the even and odd nonresidues $\pm\pi/2$ instead of $\pm\pi$.

In this restricted case of Fermat primes, $[x<p]$ is $1$ on the domain $0\le x<2^{2^n}$.
Thus the computation can be further simplified by leaving out the computation of $[x<p]$.

Let $g$ be the Jacobi symbol indicator function as follows for algorithm~\ref{fermat_qnr}.

\begin{align*}                                                                                                             \label{indicator}
  g(x)&=\left[\jacobi{x}{2^{2^n}+1}=-1\right] \text{ where } 0\le x<2^{2^n}\\
      &=
  \begin{cases}
  1&\text{if }x\text{ is a quadratic nonresidue of }F_n,\\
  0&\text{otherwise}.\\
  \end{cases}
\end{align*}

\begin{algorithm}
  \caption{
    Polynomial time algorithm for uniformly sampling quadratic nonresidues of Fermat primes $F_n=2^{2^n}+1$ where $n>1$.
}
  \label{fermat_qnr}
  \begin{noindqlist}
  \item Initial superposition via Hadamard.
    \[\frac{1}{{2^{2^{n-1}}}} \sum_{x=0}^{2^{2^n}-1} \ket{x}\]
  \item Compute Jacobi symbol indicator.
  \[\frac{1}{{2^{2^{n-1}}}} \sum_{x=0}^{2^{2^n}-1} \ket{x}\ket{g(x)}\]
\item Controlled $Z$ gate on the low bit of $x$ and $g(x)$, which multiplies all odd QNRs by $-1$.
  \[\frac{1}{{2^{2^{n-1}}}} \sum_{x=0}^{2^{2^n}-1} (-1)^{x_0g(x)}\ket{x}\ket{g(x)}.\]
\item Perform an $S$ gate on $\ket{g(x)}$, which multiplies all QNRs by $i$.
  \[\frac{1}{{2^{2^{n-1}}}} \sum_{x=0}^{2^{2^n}-1} i^{2x_0g(x)}i^{g(x)}\ket{x}\ket{g(x)}.\]
\item Uncompute the indicator function $g$.
\[\frac{1}{{2^{2^{n-1}}}} \sum_{x=0}^{2^{2^n}-1} i^{(2x_0+1)g(x)}\ket{x}.\]
\item Perform a Hadamard transform, negate $\ket{0}$ and then perform another Hadamard transform. 
This Grover step inverts all states about the mean.
  \[\frac{1}{{2^{2^{n-1}}}} \sum_{x=0}^{2^{2^n}-1} \left(1-i^{(2x_0+1)g(x)}\right)\ket{x}.\]
\item Observe a quadratic nonresidue of $F_n$.
  \end{noindqlist}

\end{algorithm}

Note that at the conclusion of algorithm~\ref{fermat_qnr}, the probability amplitude for a state $\ket{x}$ is zero when $g(x)=0$. 
The observation samples uniformly among all quadratic nonresidues.

The Fermat primes also share another property that may be useful for minimizing the required NISQ computation.
Since exactly half of the values less than $2^{2^n}$ are quadratic nonresidues, the indicator function for quadratic nonresidues is balanced.
This means that there is a permutation $\sigma\colon \Z_2^n \to \Z_2^n$ where one of the coordinate functions of $\sigma$ is the indicator function for quadratic nonresidues.

\section{The First Circuit: Finding Quadratic Nonresidues of 17}

The easiest QNR test for a NISQ will be for the prime $17$. 
The goal of this section is to design the easiest challenge problem for a NISQ device.
We use algorithm~\ref{fermat_qnr} to construct figure~\ref{fig:full_circuit} which shows a circuit for the $p=17$ Fermat quadratic nonresidue algorithm.
The circuit $\sigma$ computes the indicator function $\ii$ on qubit $x_1$ and leaves the parity bit $x_0$ alone.

To run the Fermat version of the QNR quantum algorithm we need two subcircuits:
\begin{itemize}
  \item A circuit $\sigma$ that leaves the parity bit alone and computes the Jacobi symbol $\iii$ on the second qubit.
  \item A circuit that negates the phase of the zero state, $\ket{0000}\mapsto -\ket{0000}$.
\end{itemize}

\begin{figure*}[h!]
  \includegraphics[width=\textwidth]{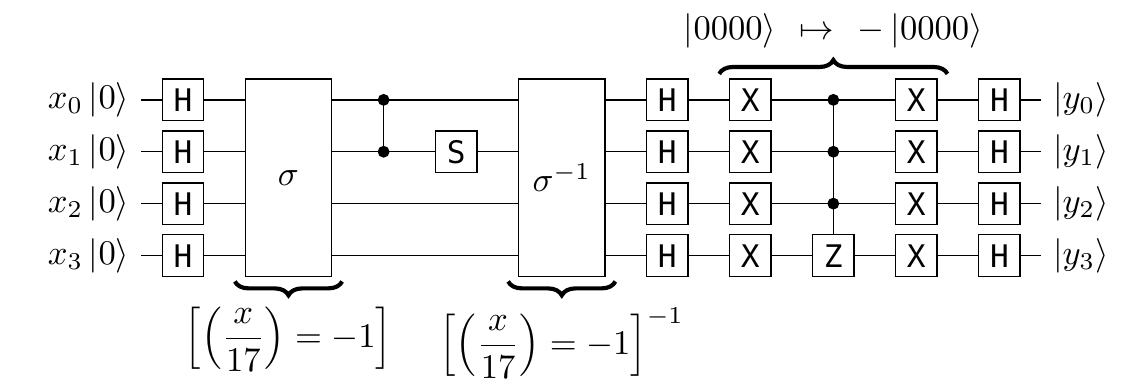}
     \caption{Full circuit for quadratic nonresidues mod $17$.} \label{fig:full_circuit}
\end{figure*}
In designing a test circuit, we should be aware of the limitations of current NISQs.
Honeywell, IBM, IonQ, and Rigetti currently provide public access to their quantum hardware.
We would like to design a circuit that runs on all of them.
The following circuit criteria performs well on all four platforms: 
\begin{itemize}
  \item Decompose all multi-qubit gates except for \CNOT and \CZ.
  \item \CNOT and \CZ gates must be nearest neighbor.
  \item Allow arbitrary phase rotations in circuit.
  \item Prefer minimizing scratch space to reducing gate count.
\end{itemize}

\subsection{Computing $\ii$}
\begin{figure*}[h]
   \centering{
     \includegraphics[width=\textwidth]{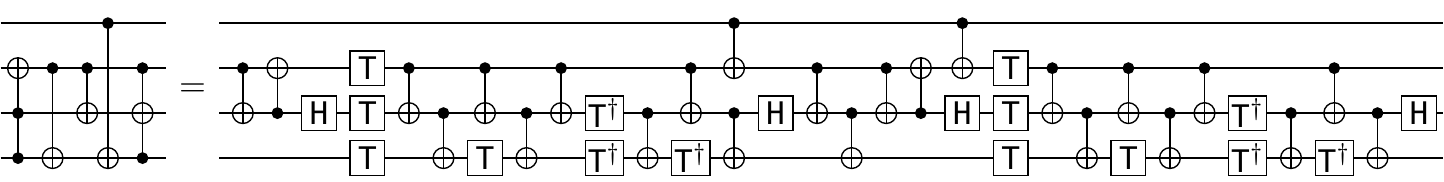}}
     \caption{Nearest neighbor $T$-gate decomposition of $\iii^{-1}$.}
  \label{fig:indicator_T_circuit}
\end{figure*}

Let
$\qq=\left\{x \,\middle|\, 0\le x <17\text{ and }\jacobi{x}{17}=-1\right\}$.
Note that $17\equiv 1\bmod{4}$ implies $-1$ is a quadratic residue of $17$, and thus $\hbox{max}(\qq)<16$.
Since $|\qq|=8$ and $\text{max}(\qq)<16$, the indicator function $[x\in \qq]=\iii$ restricted to four bits is a balanced function, illustrated in table~\ref{table:tt17}.
We can avoid using an extra bit of scratch space by computing $\iii$ in place as part of a permutation.
There are many $4$-bit permutations that will have $\iii$ as a coordinate function. Our goal is to find the simplest such permutation.

\begin{figure*}[h!]
   \centering{
     \includegraphics[width=\textwidth]{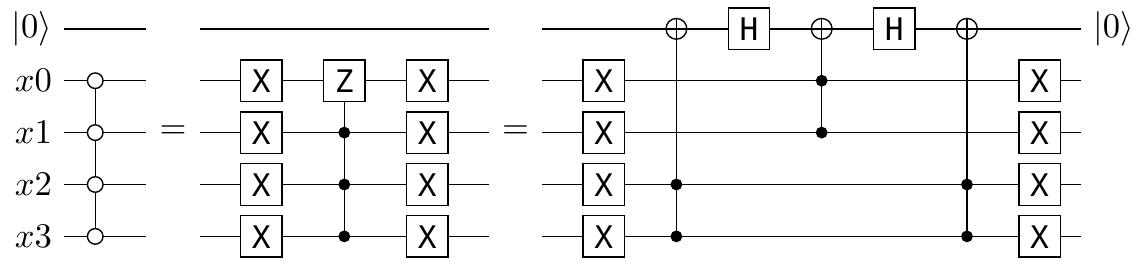}}
     \caption{Toffoli and \CCZ decomposition of 4-bit phase flip of $\ket{0000}$.}
  \label{fig:basiczeroflip}
\end{figure*}

\begin{table}[h!]
  \caption{Truth table of quadratic nonresidues mod $17$.}
  \label{table:tt17}

\[\begin{tabular}{cccc@{\hspace{1cm}}c@{\hspace{1cm}}c}
    \toprule
    $x_3$&$x_2$&$x_1$&$x_0$&$x$&$\ii$\\
    \midrule
    0&0&0&0&0&0\\
    0&0&0&1&1&0\\
    0&0&1&0&2&0\\
    0&0&1&1&3&1\\
    0&1&0&0&4&0\\
    0&1&0&1&5&1\\
    0&1&1&0&6&1\\
    0&1&1&1&7&1\\
    1&0&0&0&8&0\\
    1&0&0&1&9&0\\
    1&0&1&0&10&1\\
    1&0&1&1&11&1\\
    1&1&0&0&12&1\\
    1&1&0&1&13&0\\
    1&1&1&0&14&1\\
    1&1&1&1&15&0\\
    \bottomrule
\end{tabular}\]

\end{table}

As a bit polynomial,
$\iii=x_0x_1+x_0x_2+x_1x_2+x_1x_3+x_2x_3+x_0x_1x_3$.
Since the indicator function has a cubic term, it requires at least two Toffoli
gates. Figure~\ref{fig:indicator_circuit} shows a permutation
circuit\footnote{Circuit diagrams created with \qpic\cite{dk}} using only two
Toffoli gates.

\begin{figure}[h!]
  \centering
    {\includegraphics[width=2.5in]{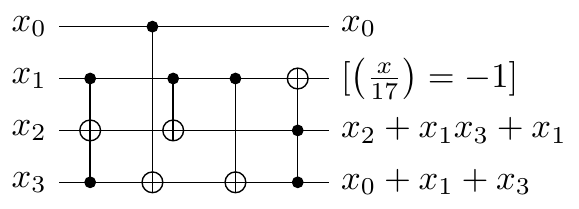}}
    \caption{An indicator permutation for $\ii$.}
    \label{fig:indicator_circuit}
\end{figure}

Note that the parity wire $x_0$ was left unchanged.
The phases of odd quadratic nonresidues can now be flipped by a controlled-$Z$ gate between the $x_0$ and $x_1$ wires.

Since we cannot have Toffoli gates in our final circuit, we search for a nearest neighbor decomposition of the circuit.
As we will see shortly, in this particular case, we only need the circuit for $\iii^{-1}$.
Figure~\ref{fig:indicator_T_circuit} was found using a computer search and the nearest neighbor decomposition of a Toffoli gate from \cite{3957}.  

\subsection{Negating the phase of the $\ket{0000}$ (a.k.a the Grover gate)}
\begin{figure*}[h!]
   \centering{
     \includegraphics[width=\textwidth]{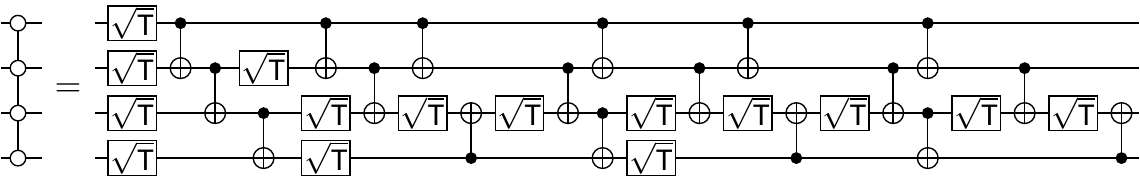}}
     \caption{Nearest neighbor $\sqrt{T}$-gate decomposition of 4-bit phase flip of $\ket{0}$.}
  \label{fig:zeroflip}
\end{figure*}

One major step of Grover's algorithm is changing the sign on the all zero state.
We will denote the action of $\ket{0}^{\otimes n} \mapsto -\ket{0}^{\otimes n}$ by a gate with all open circles and refer to is as the {\bf Grover gate}.
One possible way of negating the phase of $\ket{0000}$ is to negate the entire input and decompose the resulting \csword{CCCZ} gate.
Figure~\ref{fig:basiczeroflip} show a decomposition of a 4-qubit Grover gate into a single \csword{CCCZ} gate, which in turn is decomposed into a \CCZ gate, two Toffoli gates, and uses a scratch qubit.

Avoiding the use of extra qubits typically produces better result, and the computation of the Grover gate is the only part of the circuit that requires an additional qubit.
Thus, we seek to negate the zero state without using any scratch space.
Since all of the current NISQs allow for arbitrary phase rotation, we can do this using $\sqrt{\csword{T}}$ gates ($\pi/16$ rotations).

Recall the elementary symmetric functions
\[\sigma_j(x_1,\ldots,x_n) = \sum_{1\le i_1 < \cdots < i_j\le n}\prod_{k=1}^j x_{i_k}, \]
where $\sigma_0=1$ and $\sigma_j=0$ for $j<0$ or $j>n$.
The variables $\sigma_j$ is acting on should be clear from context.



In the decomposition of a \CCZ gate into \CNOT and \xT gates, $\pi/8$ rotations are applied to various linear combinations so that the accumulation of the rotations is zero except for $\ket{111}$.
We will use the same approach to negate the $\ket{0000}$ state using $\sqrt{\xT}$ gates.
Using the standard decomposition of $\oplus$ (\XOR) in the integers, we derive the following relationship between linear combinations and sums of symmetric polynomials.

\begin{align*}
  a\oplus b &= a+b-2ab=\sigma_1-2\sigma_2\\
  a\oplus b\oplus c &= \sigma_1-2\sigma_2+4\sigma_3\\
  a\oplus b\oplus c\oplus d &= \sigma_1-2\sigma_2+4\sigma_3-8\sigma_4\\
\end{align*}

When a phase rotation is applied to a linear combination, the result can be expressed as a root of unity raised to the polynomial associated with the linear combination.
Summing over each set of nonzero linear combinations by Hamming weight yields

\begin{align*}
  \sum_{i} x_i &= \sigma_1\\
  \sum_{i< j} x_i\oplus x_j &= 3\sigma_1 -2\sigma_2\\
  \sum_{i< j< k} x_i\oplus x_j\oplus x_k &= 3\sigma_1 -4\sigma_2+4\sigma_3\\
  x_1\oplus x_2\oplus x_3\oplus x_4 &= \sigma_1-2\sigma_2+4\sigma_3-8\sigma_4\\
\end{align*}

with a total sum of
\[ 8(\sigma_1-\sigma_2+\sigma_3-\sigma_4). \]

Consider the product of $x_i+1$ on four variables,
\[ \prod_{i=1}^4 \left(x_i+1\right)=1+\sigma_1+\sigma_2+\sigma_3+\sigma_4.\]

For boolean inputs, the result is odd if and only if all $x_i=0$.
Note that 
\[ 8(\sigma_1+\sigma_2+\sigma_3+\sigma_4)\equiv 8(\sigma_1-\sigma_2+\sigma_3-\sigma_4) \bmod{16}, \]
and the value is $8\pmod{16}$ for all nonzero boolean input and $0\pmod{16}$ for the all zero input. 

Thus, applying a phase rotation of $\pi/16$ to every nonzero linear combination results in a $-1$ phase being applied to all vectors except $\ket{0000}$, which is global phase equivalent to just negating the phase on $\ket{0000}$.

Figure~\ref{fig:zeroflip} shows a nearest neighbor decomposition that negates $\ket{0000}$.
An exhaustive search shows that 18 \CNOTs is minimal for nearest neighbor circuits and 14 \CNOTs for circuits without the nearest neighbor restriction.



\subsection{Final circuit}
\begin{figure*}[h!]
   \centering{
     \includegraphics[width=\textwidth]{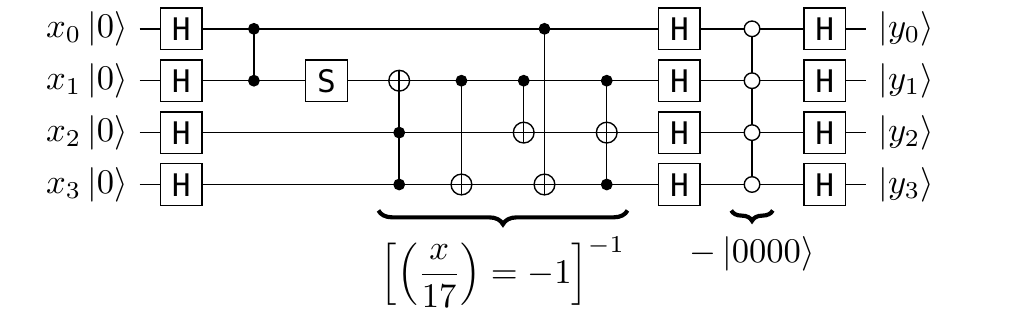}}
     \caption{Reduced circuit for quadratic nonresidues mod $17$.}
  \label{fig:circuit}
\end{figure*}

One surprising optimization that can be done immediately is the removal of the forward path of the $\iii$ permutation.
After the initial Hadamard transform, all of the $x_i$ qubits are in the $\ket{+}$ state.
Any Toffoli or \CNOT gate that targets a $\ket{+}$ gate does nothing, since \NOT applied to $\ket{+}$ is still $\ket{+}$.
Thus we only have to uncompute the $\iii$ permutation. Figure~\ref{fig:circuit} shows the simplified circuit.

The circuit in figure~\ref{fig:circuit} can likely be simplified further, but the circuit should already be within the realm of something a very basic NISQ could execute.
This particular circuit will be used to compare NISQ devices from Honeywell, IBM, IonQ, and Rigetti.

\section{The \Ppoly Problem}
When writing small programs for a quantum computer, there is a tendency to spend a lot of time optimizing the circuit. Finding the simple circuit above for $\qq$ is an example of that. In general, this changes the complexity class for a polynomial time algorithm from \xP to \Ppoly. \Ppoly contains algorithms that run in polynomial time that have received polynomial-sized advice. Finding a quadratic nonresidue lies in randomized polynomial time, \RP, which is contained in \Ppoly~\cite{Adleman1978}.

\begin{algorithm}
  \caption{(Classical) Polynomial time algorithm for finding random QNRs with advice and a
    random source.}
  \label{rand_advice}
  \begin{noindlist}
  \item Let $a=f(p)$ be a quadratic nonresidue provided by advice function.
  \item Choose a random $r\in 1,2,\ldots ,p-1$.
  \item Return $ar^2\bmod p$.
  \end{noindlist}

\end{algorithm}

Algorithm~\ref{rand_advice} samples uniformly from the quadratic nonresidues modulo $p$.
$\Ppoly$ doesn't put a bound on how much work was done to create the function that provides the advice.
So to be fair, we should tread lightly when it comes to how much work we allow in the optimization of a particular quantum circuit if we are not going to allow similar preprocessing work on the classical side.

Ideally, a fair quantum QNR test should have the following properties:
\begin{enumerate}
  \item The initial circuit should be derived from a general algorithm like algorithm~\ref{quantum_qnr}.
  \item All circuit optimizations should be drawn from an established set of general optimizations.
  \item Any circuit translation to gates particular to a specific NISQ device should also be derived from an established set of general conversions for that NISQ.
  \item A classical computer should not be able to use any part of the quantum circuit to find quadratic nonresidues in polynomial time.
\end{enumerate}

For $p=17$, the exhaustive search to find a permutation that calculated $\iii$ as one of the output wires inadvertently encoded what it meant to be a quadratic nonresidue for $p=17$.
Consider figure~\ref{fig:reverse_circuit}, which is the indicator circuit (figure~\ref{fig:indicator_circuit}) reversed.
Any three bits $a,b,c$ with a $1$ on the second wire will produce a quadratic nonresidue. Thus, this circuit classically solves the problem for $p=17$.

\begin{figure}[h!]
  \centering
    {\includegraphics[width=2.5in]{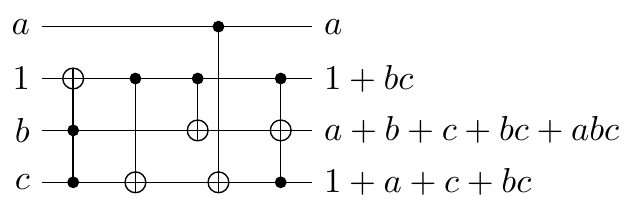}}
     \caption{Reversed circuit that generates quadratic nonresidues mod $17$.}
    \label{fig:reverse_circuit}
\end{figure}

\begin{figure*}[h]
   \centering{
     \includesvg[width=\textwidth]{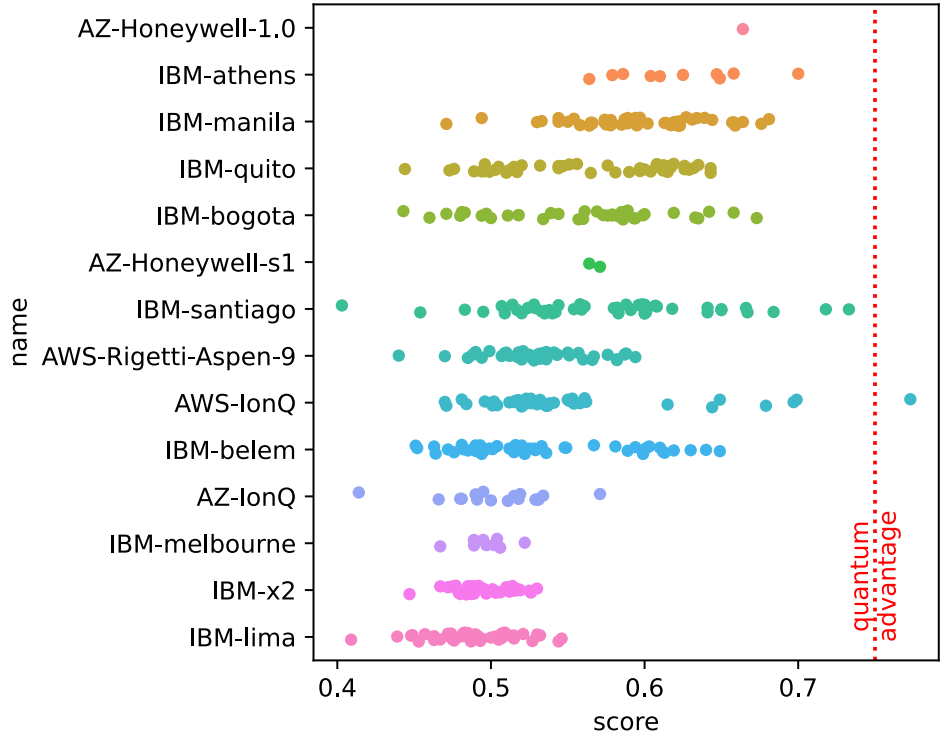}
     \caption{Success rate of 1000 shot runs on different devices, sorted by median. (June 19-Aug 31, 2021)}
   }
  \label{fig:nisq_score}
\end{figure*}

\begin{figure*}[h!]
   \centering{
     \includesvg[width=\textwidth]{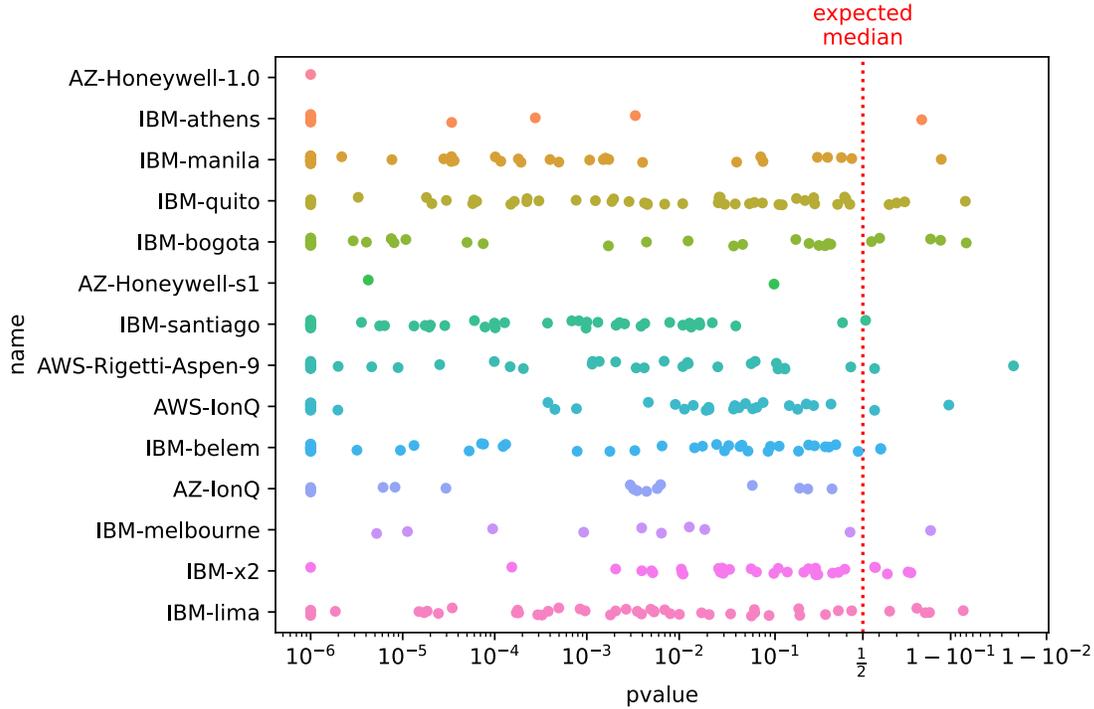}
     \caption{$p$-value of 1000 shot runs on different devices, sorted by median. (June 19-Aug 31, 2021)}
   }
  \label{fig:pvalue}
\end{figure*}

\begin{figure*}[h!]
   \centering{
     \includesvg[width=\textwidth]{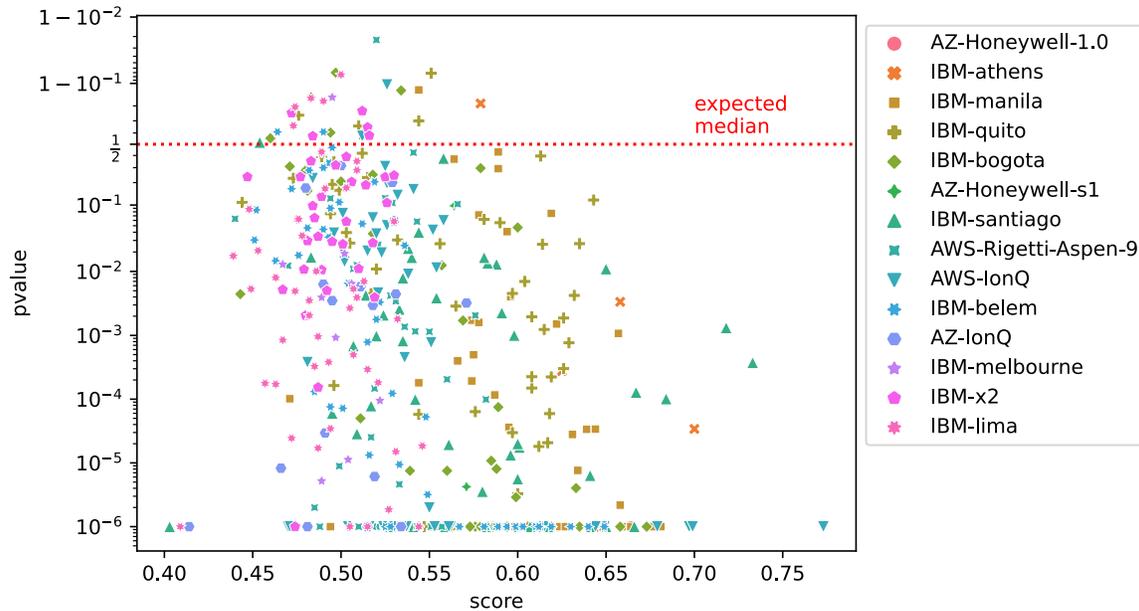}
     \caption{Scatter plot success rate and $p$-value per run of 1000 shots. Legend sorted by median or scores. (June 19-Aug 31, 2021)}
   }
  \label{fig:scatterplot}
\end{figure*}
For this reason, the circuit of figure~\ref{fig:circuit} fails property 4 of our fairness requirements.
In fact, any 4-bit permutation that calculates $\iii$ on an output wire, can be used to generate quadratic nonresidues for $p=17$ using exactly the same trick.

This trick does not solve the problem of finding quadratic nonresidues in general.
We expect algorithms that calculate the Jacobi symbol for arbitrary $p\equiv 1\bmod 8$ to produce circuits that require scratch space.
Those circuits cannot be trivially reversed since we don't know what the value of the internal scratch space will be when the Jacobi symbol is calculated.

\section{Scoring Current NISQ Devices}
Acknowledging the unfairness of the circuit in figure~\ref{fig:circuit}, we nevertheless are interested in how well current NISQs perform this new test.
A fair circuit for $p=17$ will likely require more gates and qubits.
Thus, while not ideal, the unfairly optimized circuit will provide a useful lower bound to get things started.
As we will we see, even this very simple circuit is a real challenge for current NISQs.
The data was collected between June 19, 2021 and August 31, 2021.
Table~\ref{table:runs} shows the number of successful collection days for each device.
Not all devices had equal availability or ease of use to run.

\begin{table}[h!]
  \caption{Number of runs completed per device}
  \label{table:runs}

  \[\begin{tabular}{lc}
    \toprule
    NISQ & Completed runs\\
    \midrule
IBM-lima & 47 \\
IBM-quito & 49 \\
IBM-manila & 49 \\
IBM-belem & 49 \\
IBM-bogota & 38 \\
IBM-athens & 10 \\
IBM-santiago & 49 \\
IBM-melbourne & 10 \\
IBM-x2 & 33 \\
AZ-IonQ & 17 \\
AZ-Honeywell-s1 & 2 \\
AZ-Honeywell-1.0 & 1 \\
AWS-IonQ & 43 \\
AWS-Rigetti-Aspen-9 & 43 \\
    \bottomrule
\end{tabular}\]
\end{table}

Figure 12 shows the success rate of $1000$ shot runs for various NISQs during experimentation in the summer of 2021.
Architectures such as IonQ could benefit from slightly smaller circuits since they do not have a nearest neighbor requirement, but for this initial comparison, all NISQs were purposely given the same circuit.

Each run consists of $1000$ shots, and the score is what percentage of those observations were quadratic nonresidues.
The {\bf quantum advantage} line indicates $75\%$ correct, beyond which a quantum computer starts providing evidence of something we don't know how to do classically.
The NISQs were tested via three different platforms, IBM, Microsoft Azure (AZ), and Amazon (AWS).


Since the success rate of a run seemed to be fairly consistent on a given day, at most one run per day per device was done.
When more longitudinal data is available, it will be interesting to see if the scores degrade over time, or if they stay relatively constant.

Surprisingly, some of the NISQ devices were worse than random, even for multiple runs.
Since $\ket{0000}$ and $\ket{1111}$ are both quadratic residues, if a NISQ failed in a biased way towards $\ket{0}$ or $\ket{1}$, it could easily be causally worse than random.

\subsection{Honeywell tests}
Unfortunately, Honeywell was in the midst of retiring their "1.0" system and bringing up their new "s1" system during the testing period.
Frustratingly, the strong single score for the "1.0" system begs for more data.
The new "s1" system seems to not perform as well as the "1.0" system, but the data here is lacking to make any conclusive judgment.

\section{Succeeding Uniformly}

It is not enough for a NISQ to produce quadratic nonresidues at a high success rate, but it should do so uniformly.
If a NISQ has a particular failure bias, and that bias favors a particular subset of QNRs, this should be considered lucky and not good.
For each run of $1000$ shots, we calculate the $p$-value based on the Chi Square test.
Given a $1000$ shot run, the number of QNR observations will be $k=1000*score$.
Let $O_i$ be the number of observations of the $i$th QNR.
Since their are $8$ QNRs modulo $17$, the expected number of each observed QNR is $E(O_i)=E_i=k/8$.
We then calculate the chi-square statistic for our actual observations.
\[ \chi^2 = \sum_{i=1}^8 \frac{(O_i-k/8)^2}{k/8} \]
Since we have $8$ QNR observations, this will be a chi-square statistic with $7$ degrees of freedom.
Following standard practice, we compute a $p$-value of our observation, which is the probability that we would see our observation or something more extreme given the assumption that we are sampling from a uniform distribution on $8$ elements.

With the uniformity assumption, our $p$-values will be uniformly distributed from $0$ to $1$.
As we see in figure 13, our samples are very rough, indicating a strong lack of uniformity.
Since the area we are interested in is concentrated around $\frac{1}{2}$, $p$-values less that $10^{-6}$ are simply truncated to that level.

A NISQ device running the QNR17 test correctly should not favor any particular QNR or subset of QNRs over the others.
From this point of view, the single promising result from the Honeywell "1.0" system looks suspect.
Since the $p$-value for uniformity is less than $1$ in a million, it seems that the system more likely happened to be failing into a QNR state than running the algorithm with high fidelity.

For the higher scoring runs, we would like to see increased evidence of uniformity.
Instead we see the opposite.
Figure 14 shows a significant decrease in the likelihood of uniformity for high scoring runs on the given NISQs.
Up to scores of $.60$, there are still $p$-values close to $\frac{1}{2}$.
After that point, the likelihood of any run being uniform is very low.

Additionally, it should be noted that nearly half of the runs had a $p$-value of less than $10^{-6}$.
This indicates that however the devices were failing, they were failing in a biased way with respect to the QNRs.

\section{Conclusion}
We look forward to additional NISQ devices running the algorithm from figure~\ref{fig:circuit} for $p=17$.
Currently, most results are barely above the noise threshold, and are not yet exceeding 75\% regularly.
Of the runs that do score well, most have a very low $p$-value, so they might be succeeding for a causally bad reason.
QASM code is included in the appendix for trying out the QNR17 test.

As NISQs improve, algorithms for $p=41$, $F_4=257$, and $F_5=65537$ will likely be of interest.
For initial evaluations, it is reasonable to use small, unfairly optimized circuits, but we expect to see fair circuits derived from the full algorithm and reduced by a rules-based optimizer as NISQ devices become more capable.

\section{Acknowledgments}
Thank you to Dmitri Maslov, John Preskill, Anthony Gamst, and Thomas L. Draper who all provided helpful feedback on early drafts of this paper. 

\bibliographystyle{IEEEtran}  
\bibliography{IEEEabrv,QNR_NISQ}  

\begin{thebibliography}{10}
\providecommand{\url}[1]{#1}
\csname url@samestyle\endcsname
\providecommand{\newblock}{\relax}
\providecommand{\bibinfo}[2]{#2}
\providecommand{\BIBentrySTDinterwordspacing}{\spaceskip=0pt\relax}
\providecommand{\BIBentryALTinterwordstretchfactor}{4}
\providecommand{\BIBentryALTinterwordspacing}{\spaceskip=\fontdimen2\font plus
\BIBentryALTinterwordstretchfactor\fontdimen3\font minus
  \fontdimen4\font\relax}
\providecommand{\BIBforeignlanguage}[2]{{%
\expandafter\ifx\csname l@#1\endcsname\relax
\typeout{** WARNING: IEEEtran.bst: No hyphenation pattern has been}%
\typeout{** loaded for the language `#1'. Using the pattern for}%
\typeout{** the default language instead.}%
\else
\language=\csname l@#1\endcsname
\fi
#2}}
\providecommand{\BIBdecl}{\relax}
\BIBdecl

\bibitem{Centrone_2021}
\BIBentryALTinterwordspacing
F.~Centrone, N.~Kumar, E.~Diamanti, and I.~Kerenidis, ``Experimental
  demonstration of quantum advantage for {NP} verification with limited
  information,'' \emph{Nature Communications}, vol.~12, no.~1, Feb 2021.
  [Online]. Available: \url{http://dx.doi.org/10.1038/s41467-021-21119-1}
\BIBentrySTDinterwordspacing

\bibitem{Kumar_2019}
\BIBentryALTinterwordspacing
N.~Kumar, I.~Kerenidis, and E.~Diamanti, ``Experimental demonstration of
  quantum advantage for one-way communication complexity surpassing best-known
  classical protocol,'' \emph{Nature Communications}, vol.~10, no.~1, Sep 2019.
  [Online]. Available: \url{http://dx.doi.org/10.1038/s41467-019-12139-z}
\BIBentrySTDinterwordspacing

\bibitem{Arrazola_2018}
\BIBentryALTinterwordspacing
J.~M. Arrazola, E.~Diamanti, and I.~Kerenidis, ``Quantum superiority for
  verifying np-complete problems with linear optics,'' \emph{npj Quantum
  Information}, vol.~4, no.~1, Nov 2018. [Online]. Available:
  \url{http://dx.doi.org/10.1038/s41534-018-0103-1}
\BIBentrySTDinterwordspacing

\bibitem{Bravyi_2018}
\BIBentryALTinterwordspacing
S.~Bravyi, D.~Gosset, and R.~König, ``Quantum advantage with shallow
  circuits,'' \emph{Science}, vol. 362, no. 6412, p. 308–311, Oct 2018.
  [Online]. Available: \url{http://dx.doi.org/10.1126/science.aar3106}
\BIBentrySTDinterwordspacing

\bibitem{aaronson2016complexitytheoretic}
S.~Aaronson and L.~Chen, ``Complexity-theoretic foundations of quantum
  supremacy experiments,'' 2016.

\bibitem{46227}
\BIBentryALTinterwordspacing
S.~Boixo, S.~Isakov, V.~Smelyanskiy, R.~Babbush, N.~Ding, Z.~Jiang, M.~J.
  Bremner, J.~Martinis, and H.~Neven, ``Characterizing quantum supremacy in
  near-term devices,'' \emph{Nature Physics}, vol.~14, p. 595–600, 2018.
  [Online]. Available: \url{https://www.nature.com/articles/s41567-018-0124-x}
\BIBentrySTDinterwordspacing

\bibitem{gauss_qnr}
C.~F. Gauss, ``Disquisitones arithmeticae,'' 1801.

\bibitem{DBLP:journals/corr/abs-1004-2091}
\BIBentryALTinterwordspacing
R.~P. Brent and P.~Zimmerman, ``An {$O(M(n) \log n)$} algorithm for the jacobi
  symbol,'' \emph{CoRR}, vol. abs/1004.2091, 2010. [Online]. Available:
  \url{http://arxiv.org/abs/1004.2091}
\BIBentrySTDinterwordspacing

\bibitem{harvey:hal-02070778}
\BIBentryALTinterwordspacing
D.~Harvey and J.~Van Der~Hoeven, ``{Integer multiplication in time O(n log
  n)},'' \emph{{Annals of Mathematics}}, 2020. [Online]. Available:
  \url{https://hal.archives-ouvertes.fr/hal-02070778}
\BIBentrySTDinterwordspacing

\bibitem{1993--cohen}
H.~Cohen, \emph{A course in computational algebraic number theory}, ser.
  Graduate Texts in Mathematics.\hskip 1em plus 0.5em minus 0.4em\relax
  Spinger-Verlag, Berlin, 1993, vol. 138.

\bibitem{Grover96}
L.~K. Grover, ``A fast quantum mechanical algorithm for database search,'' in
  \emph{Proceedings of the {T}wenty-Eigth Annual {ACM} Symposium on the
  {T}heory of {C}omputing (Philadelphia, Pennsylvania, USA, 1996)}, May 22--24
  1996, pp. 212--219.

\bibitem{Bernstein97quantumcomplexity}
E.~Bernstein and U.~Vazirani, ``Quantum complexity theory,'' \emph{SIAM JOURNAL
  ON COMPUTING}, pp. 1411--1473, 1997.

\bibitem{Adleman97quantumcomputability}
L.~M. Adleman, J.~Demarrais, and M.~dh~A.~Huang, ``Quantum computability,''
  \emph{SIAM JOURNAL OF COMPUTATION}, pp. 1524--1540, 1997.

\bibitem{Mosca04exactquantum}
M.~Mosca and C.~Zalka, ``Exact quantum {F}ourier transforms and discrete
  logarithm algorithms,'' \emph{International Journal of Quantum Information},
  pp. 91--100, 2004.

\bibitem{dawson2005solovaykitaev}
C.~M. Dawson and M.~A. Nielsen, ``The {S}olovay-{K}itaev algorithm,'' 2005.

\bibitem{selinger2012efficient}
P.~Selinger, ``Efficient {C}lifford{$+T$} approximation of single-qubit
  operators,'' arXiv:quant.ph/1212.6253, Dec. 2012.

\bibitem{comparator}
D.~Oliveira and R.~Ramos, ``Quantum bit string comparator: Circuits and
  applications,'' \emph{Quantum Computers and Computing}, vol.~7, 01 2007.

\bibitem{Crandall02thetwenty-fourth}
R.~E. Crandall, E.~W. Mayer, and J.~S. Papadopoulos, ``The twenty-fourth
  {F}ermat number is composite,'' \emph{Mathematics of Computation}, vol.~72,
  no. 243, pp. 1555--1572, 2003.

\bibitem{dk}
T.~G. Draper and S.~A. Kutin, ``{$\langle\textrm{q}|\textrm{pic}\rangle$:}
  {Q}uantum circuits diagrams in {\LaTeX},''
  \emph{http://github.com/qpic/qpic}.

\bibitem{3957}
C.~Gidney, ``Minimum number of cnots for toffoli with non-adjacent controls
  (answer:3964),'' \emph{https://quantumcomputing.stackexchange.com}.

\bibitem{Adleman1978}
L.~Adleman, ``Two theorems on random polynomial time,'' in \emph{19th {A}nnual
  {S}ymposium on {F}oundations of {C}omputer {S}cience}.\hskip 1em plus 0.5em
  minus 0.4em\relax IEEE, 1978, pp. 75--83.

\end{thebibliography}

\newpage
\appendix
\subsection*{QASM source code}
\tt
\begin{lstlisting}
OPENQASM 2.0;
include "qelib1.inc";

qreg q[4];
creg c[4];

# Some hadamard gates can be
# cancelled in the beginning.
# h q[0];
h q[1];
h q[2];
h q[3];
# h q[0];
cx q[1],q[0];
h q[0];
s q[1];
cx q[1],q[2];
cx q[2],q[1];
h q[2];
t q[1];
t q[2];
t q[3];
cx q[1],q[2];
cx q[2],q[3];
t q[3];
cx q[1],q[2];
cx q[2],q[3];
tdg q[3];
cx q[1],q[2];
tdg q[2];
cx q[2],q[3];
tdg q[3];
cx q[1],q[2];
cx q[2],q[3];
cx q[0],q[1];
h q[2];
cx q[1],q[2];
cx q[2],q[3];
cx q[1],q[2];
cx q[2],q[1];
cx q[0],q[1];
h q[2];
t q[1];
t q[2];
t q[3];
cx q[1],q[2];
cx q[2],q[3];
t q[3];
cx q[1],q[2];
cx q[2],q[3];
tdg q[3];
cx q[1],q[2];
tdg q[2];
cx q[2],q[3];
tdg q[3];
cx q[1],q[2];
cx q[2],q[3];
h q[2];
h q[0];
h q[1];
h q[2];
h q[3];
rz(pi/8) q[0];
rz(pi/8) q[1];
rz(pi/8) q[2];
rz(pi/8) q[3];
cx q[0],q[1];
cx q[1],q[2];
cx q[2],q[3];
rz(pi/8) q[1];
rz(pi/8) q[2];
rz(pi/8) q[3];
cx q[0],q[1];
cx q[1],q[2];
rz(pi/8) q[2];
cx q[3],q[2];
rz(pi/8) q[2];
cx q[0],q[1];
cx q[1],q[2];
cx q[2],q[3];
rz(pi/8) q[2];
rz(pi/8) q[3];
cx q[0],q[1];
cx q[1],q[2];
rz(pi/8) q[2];
cx q[3],q[2];
rz(pi/8) q[2];
cx q[0],q[1];
cx q[1],q[2];
cx q[2],q[3];
rz(pi/8) q[2];
cx q[0],q[1];
cx q[1],q[2];
rz(pi/8) q[2];
cx q[3],q[2];
h q[0];
h q[1];
h q[2];
h q[3];
measure q[0] -> c[0];
measure q[1] -> c[1];
measure q[2] -> c[2];
measure q[3] -> c[3];
\end{lstlisting}

\end{document}